\documentclass[onefignum,onetabnum]{siamart171218} 

\usepackage{amsfonts}
\usepackage{graphicx}
\usepackage{epstopdf}
\usepackage{algorithmic}
\usepackage{amssymb}
\usepackage{mathtools}
\usepackage{color}
\usepackage[shortlabels]{enumitem}
\usepackage{multirow}
\usepackage{xfrac}
\usepackage[super]{nth}
\usepackage{cite}
\usepackage{icomma}
\usepackage{array}
\usepackage{tabularx}
\usepackage{arydshln}
\usepackage{hyperref}

\allowdisplaybreaks

\ifpdf
  \DeclareGraphicsExtensions{.eps,.pdf,.png,.jpg}
\else
  \DeclareGraphicsExtensions{.eps}
\fi

\newcommand{\ZH}[3]{\ensuremath{\mathcal{Z}_{\{#1\}}^{#2 \times #2}(#3)}}
\newcommand{\HH}[3]{\ensuremath{\mathcal{H}_{\{#1\}}^{#2 \times #2}(#3)}}
\renewcommand{\proofbox}{$\natural$}

\newsiamremark{remark}{Remark}
\newsiamthm{conjecture}{Conjecture}

\headers{Bohemian Upper Hessenberg Matrices}{E. Y. S. Chan, et al.}

\title{Bohemian Upper Hessenberg Matrices}

\author{
  Eunice Y. S. Chan\thanks{Department of Applied Mathematics, Western University
    (\email{echan295@uwo.ca}, 
    \email{rcorless@uwo.ca},
    \email{sthornt7@uwo.ca}).}
  \and
  Robert M. Corless\footnotemark[1]
  \and
  Laureano Gonzalez-Vega\thanks{Departamento de Matematicas, Estadistica y Computacion, Universidad de Cantabria
  (\email{laureano.gonzalez@unican.es}).}
  \and
  J.~Rafael Sendra\thanks{Research Group ASYNACS, Departamento de F{\'i}sica y Matem\'{a}ticas, University of Alcal{\'a}
  (\email{rafael.sendra@uah.es}).}
  \and
  Juana Sendra\thanks{Universidad Polit{\'{e}}cnica de Madrid
  (\email{jsendra@etsist.upm.es}).}
  \and
  Steven E. Thornton\footnotemark[1]
}

\usepackage{amsopn}

\begin{document}

\maketitle

\begin{abstract}
We look at Bohemian matrices, specifically those with entries from $\{-1, 0, {+1}\}$. More, we specialize the matrices to be upper Hessenberg, with subdiagonal entries $\pm1$. Many properties remain after these specializations, some of which surprised us. We find two recursive formulae for the characteristic polynomials of upper Hessenberg matrices.
Focusing on only those matrices whose characteristic polynomials have maximal height allows us to explicitly identify these polynomials and give a lower bound on their height. This bound is exponential in the order of the matrix. We count \textsl{stable} matrices, normal matrices, and neutral matrices, and tabulate the results of our experiments. We prove a theorem about the only possible kinds of normal matrices amongst a specific family of Bohemian upper Hessenberg matrices. 
\end{abstract}

\section{Introduction}
A matrix family is called \textbf{Bohemian} if its entries come from a fixed finite discrete (and hence bounded) set, usually integers. The name is a mnemonic for \textbf{Bo}unded \textbf{He}ight \textbf{M}atrix of \textbf{I}ntegers. Such populations arise in many applications (e.g.~compressed sensing) and the properties of matrices selected ``at random'' from such families are of practical and mathematical interest. For example, Tao and Vu have shown that random matrices (more specifically real symmetric random matrices in which the upper-triangular entries $\xi_{i, j}$, $i < j$ and diagonal entries $\xi_{i, i}$ are independent) have simple spectrum~\cite{tao2017random}. An overview of some of our original interest in Bohemian matrices can be found in~\cite{corless2017bohemian}.

Bohemian families have been studied for a long time, although not under that name. For instance, Olga Taussky-Todd's paper ``Matrices of Rational Integers"~\cite{taussky1960matrices} begins by saying
\begin{quote}
    ``This subject is very vast and very old. It includes all of the arithmetic theory of quadratic forms, as well as many of other classical subjects, such as latin squares and matrices with elements $+1$ or $-1$ which enter into Euler's, Sylvester's or Hadamard's famous conjectures."
\end{quote}
The paper~\cite{gear1969simple} by C.~W.~Gear is another instance. What is new here is the idea that these families are themselves interesting objects of study, and susceptible to brute-force computational experiments as well as to asymptotic analysis. These experiments have generated many conjectures, some of which we resolve in this paper.  Others remain unsolved, and are listed on the Characteristic Polynomial Database~\cite{CPDB}. Many of the conjectures have a number-theoretic or combinatorial flavour.

Typical computational puzzles arise on asking simple-looking questions such as ``how many $6 \times 6$ matrices with the population\footnote{The population of a Bohemian family is the set of permissible entries.} $\{-1, 0, {+1}\}$ are singular.'' The answer is not known as we write this, although we can give a probabilistic estimate ($0.205$ after $20,000,000$ sample determinants\footnote{4103732 singular matrices out of twenty million sampled.}): brute computation seems futile because there are $3^{36} \doteq 1.7\times10^{17}$ such matrices. We do know the answers up to size five by five: The number of $n$ by $n$ singular matrices with population $\{-1, 0, {+1}\}$ is, for $n=1$, $2$, $3$, $4$, and $5$, 
just $1$, $33$, $7,875$, $15,099,201$, and $237,634,987,683$.
This represents fractions of their numbers ($3^{n^2}$) of $0.333$, $0.407$, $0.400$, $0.351$, and $0.280$,
respectively.

Yet such matrix families  are both useful and interesting. For instance, one may use discrete optimization over a family to look for improved growth factor bounds~\cite{higham2018bohemian}. Matrices with the population $\{-1, 0, {+1}\}$ have minimal height\footnote{$\mathrm{height}(A) := ||\mathrm{vec}(A) ||_{\infty}$ is the largest absolute value of any entry in $A$.} over all integer matrices; finding a matrix in this family which has a given polynomial $p(\lambda) \in \mathbb{Z}[\lambda]$ as characteristic polynomial identifies a so-called ``minimal height companion matrix'', which may confer numerical benefits.

Recently the study of eigenvalues of structured Bohemian matrices (e.g.~tridiagonal, complex symmetric) has been undertaken and several puzzling features are seen resulting from extensive experimental computations. For instance, some of the images at \href{http://www.bohemianmatrices.com/gallery}{bohemianmatrices.com/gallery} show common features including ``holes''.

Different matrix structures produce remarkably different pictures. One structure useful in eigenvalue computation is the upper Hessenberg matrix, which means a matrix $\mathbf{H}$ such that $h_{i, j} = 0$ if $i > j + 1$. These arise naturally in eigenvalue computation because the QR iteration is cheaper for matrices in Hessenberg form. 
Results on the determinants of Hessenberg matrices can be found in \cite{kaygisiz2012determinant}.

\begin{remark}
\textit{on computing eigenvalues by first computing characteristic polynomials.}
Numerical analysts are familiar with the superior numerical stability of computing eigenvalues iteratively, usually by the QR algorithm or some variant, rather than first computing characteristic polynomials and then finding roots. As is well-known, such an algorithm is \textsl{numerically unstable} because polynomials are usually badly-conditioned while eigenvalues are usually well-conditioned\footnote{This has been well-known to the point of folklore since the work of Wilkinson.  The well-conditioning of eigenvalues has only recently been quantified in some cases, but for instance the results of~\cite{beltran2017polynomial} do confirm the folklore.}.  Somewhat surprisingly, for several families of Bohemian matrices, characteristic polynomials become valuable again: first because the matrix dimensions are typically small or at most moderate, the ill-conditioning does not matter much, and second because for some families (not all!) the number of distinct characteristic polynomials is vastly smaller than the number of matrices in the family.  For instance, for the general five by five matrices with population $\{-1, 0, {+1}\}$, there are nearly one trillion such matrices, but fewer than two million characteristic polynomials.  This compression is significant.

For other families of matrices, such as upper Hessenberg Toeplitz matrices, there is no compression at all because each matrix has a distinct characteristic polynomial.  Circulant matrices fall between, having fewer characteristic polynomials but not vastly fewer.  The lesson is that for some questions (though not others), prior computation of characteristic polynomials is valuable.
\end{remark}

We begin our study in this paper by considering determinants of Bohemian upper Hessenberg matrices. We prove two recursive formulae for the characteristic polynomials of upper Hessenberg matrices\footnote{We do not claim originality; recursion relations for upper Hessenberg determinants are known.}. For another recursive formula we refer to~\cite{elouafi2009recursion}. During the course of our computations, we encountered ``maximal polynomial height'' characteristic polynomials when the matrices were not only upper Hessenberg, but Toeplitz ($h_{i,j}$ constant along diagonals $j-i = k$); we have several results for such matrices, which will appear in \cite{chan2018BUHT}. Further restrictions to this class allowed identification of key results including explicit formulae for the characteristic polynomials of maximal height. In what follows, we lay out definitions and prove several facts of interest about characteristic polynomials and their respective height for these families.

In Figure~\ref{fig:UH_6} we see all eigenvalues of $6 \times 6$ upper Hessenberg matrices, subdiagonals fixed at $-1$, with population $P = \{-1, 0, {+1}\}$.
We denote this set of matrices $\mathcal{H}_{\{\pi\}}^{6\times 6}(P)$.  There are $3^{21}=10,460,353,203$ such matrices.  We see a wide octagonal shape.  The width of the figure reflects that some matrices might have diagonals $-1$, while some have diagonals $0$, and others have diagonals $1$.  Of course mixed diagonals are also possible, but this should only tend to push things towards the centre.

This thinking motivates considering the subset of these matrices which has diagonal fixed at $0$.  
We denote this set of matrices $\mathcal{Z}_{\{\pi\}}^{6 \times 6}(P)$. There are substantially fewer such matrices, only $3^{15}=14,348,907$ to be exact, and their eigenvalues are pictured in Figure~\ref{fig:UH_6_0_Diag}, with certain zones enhanced.  We see that, roughly speaking, Figure~\ref{fig:UH_6} is partially explained by saying that, along with other eigenvalues, it contains three copies of Figure~\ref{fig:UH_6_0_Diag} placed with centres at $-1$, at $0$, and at ${+1}$.

In this paper we seek to explain some of the features of these pictures, and to learn some things about these families of Bohemian matrices.

\begin{figure}
    \centering
    \includegraphics[width=\textwidth]{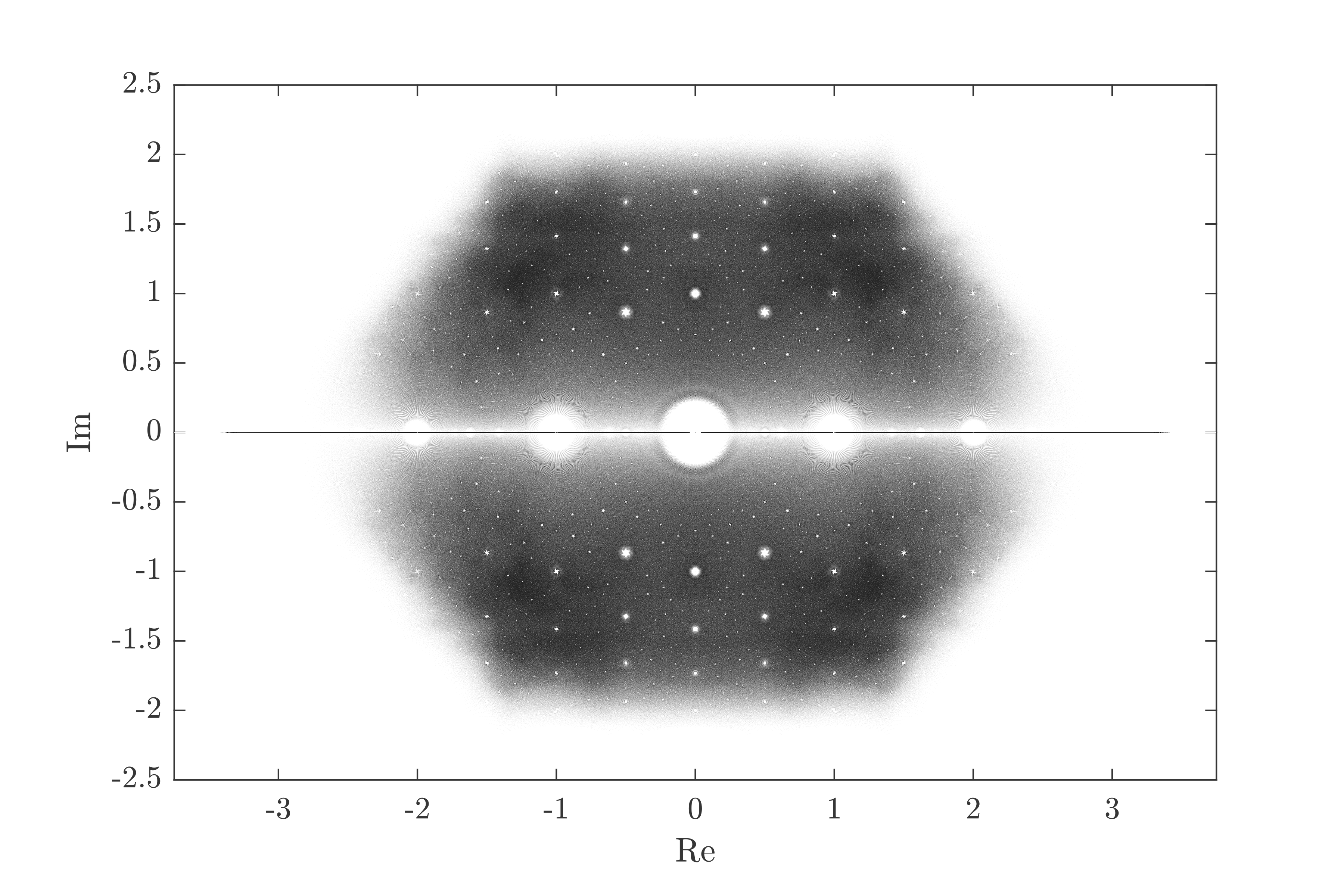}
    \caption{The set of eigenvalues of all $10,460,353,203$ six by six upper Hessenberg matrices $\mathbf{H}$  with entries $\mathbf{H}_{i,j} \in \{-1, 0, {+1}\}$, and $\mathbf{H}_{i+1,i} = -1$ for $1 \le i \le j \le 6$. A more detailed image can be found at \url{assets.bohemianmatrices.com/gallery/UH_6x6.png}}
\label{fig:UH_6}
\end{figure}

\begin{figure}
    \centering
    \includegraphics[width=\textwidth]{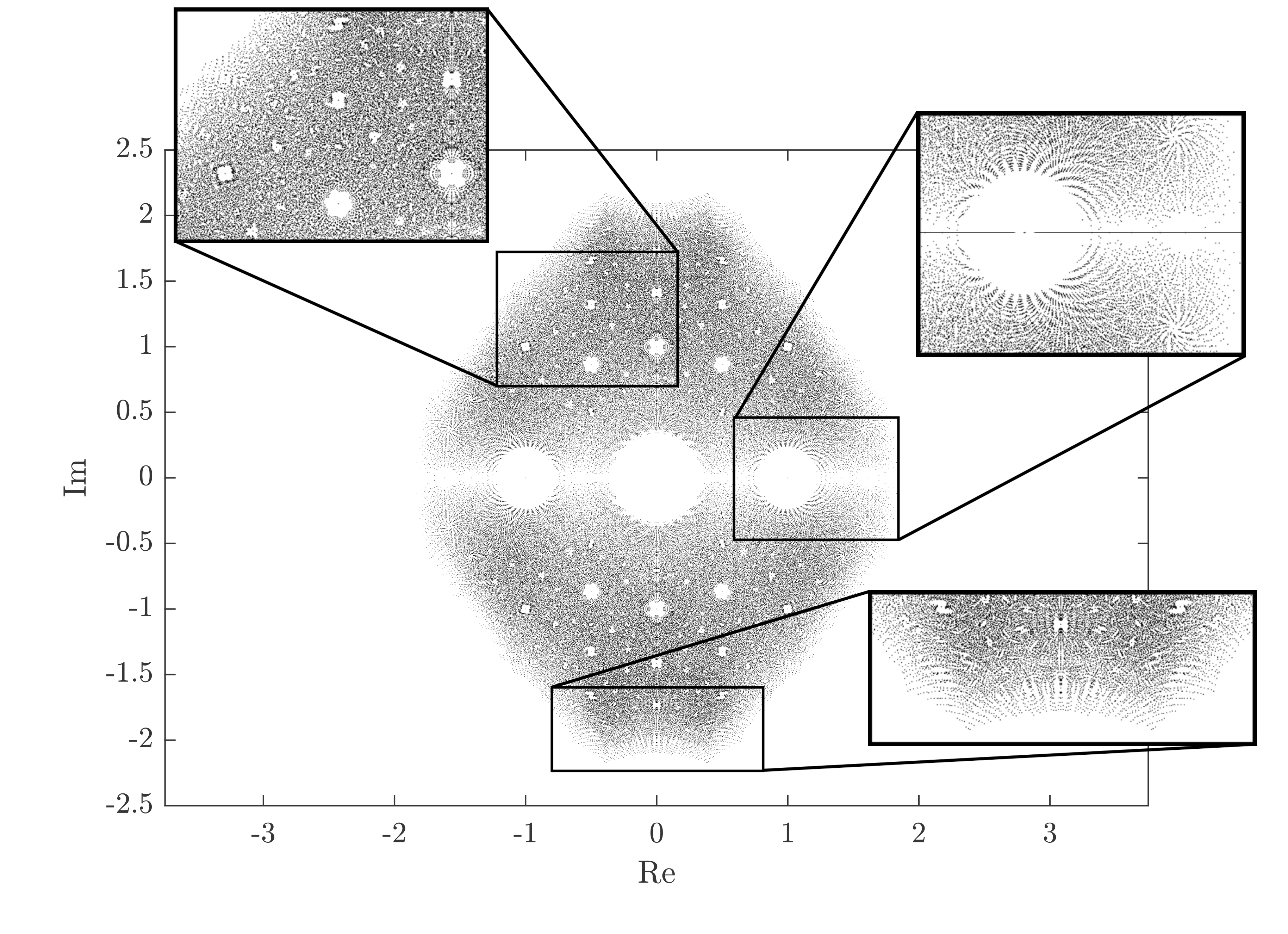}
    \caption{The set of eigenvalues of all $14,348,907$ matrices in \ZH{\pi}{6}{\{-1, 0, {+1}\}}; that is, six by six upper Hessenberg matrices $\mathbf{H}$ with entries $\mathbf{H}_{i,j} \in \{-1, 0, {+1}\}$, diagonal entries fixed as zero, and $\mathbf{H}_{i+1,i} = -1$ for $1 \le i < j \le 6$. A more detailed image can be found at \url{assets.bohemianmatrices.com/gallery/UH_0_Diag_6x6.png}}
\label{fig:UH_6_0_Diag}
\end{figure}

\section{Prior Work}
Visible features of graphs of roots and eigenvalues from structured families of polynomials and matrices have been previously studied. One well-known polynomial whose roots produce interesting pictures is the Littlewood polynomial,
\begin{equation}
    p(x) = \sum_{i = 0}^{n}a_{i}x^{i} \>,
\end{equation}
where $a_{i} = \{-1, {+1}\}$. These polynomials have been studied in \cite{baez2009beauty}, \cite{borwein2012computational}, \cite{borwein2001visible}, and \cite{borwein1997polynomials}. The image of their roots raises many questions, ranging from whether the set is (ultimately, as $n \to \infty$) a fractal and what the boundary of the set is, to questions about the holes in the image and its connection to various properties, such as degree and coefficients of the polynomial. Answers to some of these questions, particularly the ones involving the holes, have been shown to have some significance in number theory~\cite{beaucoup1998multiple}. Roots of other polynomials have also been visualized; for more, see Christensen's\footnote{\url{https://jdc.math.uwo.ca/roots/}} and J{\"o}rgenson's\footnote{\url{http://www.cecm.sfu.ca/~loki/Projects/Roots/}} web pages.

Corless used a generalization of the Littlewood polynomial (to Lagrange bases). In his paper \cite{corless2004generalized}, he gave a new kind of companion matrix for polynomials expressed in a Lagrange basis. He used generalized Littlewood polynomials as test problems for the algorithm.

``The Bohemian Eigenvalue Project" was first presented as a
poster~\cite{eccadposter2015} at the East Coast Computer Algebra Day (ECCAD) 2015. The poster focused on preliminary results and many of the questions raised when visualizing the distributions of Bohemian eigenvalues over the complex plane. In particular, the poster focused on ``eigenvalue exclusion zones'' (i.e. distinct regions within the domain of the eigenvalues where no eigenvalues exist), computational methods for visualizing eigenvalues, and some results on eigenvalue conditioning over distributions of random matrices.

In Chan's Master's thesis~\cite{chan2016comparison}, she extended Piers W.~Lawrence's construction of the companion matrix for the Mandelbrot polynomials~\cite{corless2013largest, corlessMandelbrot} to other families of polynomials, mainly the Fibonacci-Mandelbrot polynomials and the Narayana-Mandelbrot polynomials. What is relevant here about this construction is that these matrices are upper Hessenberg and contain entries from a constrained set of numbers: $\{-1, 0\}$, and therefore fall under the category of being Bohemian upper Hessenberg. Both the Fibonacci-Mandelbrot matrices and Narayana-Mandelbrot matrices are also Bohemian upper Hessenberg, but the set that the entries draw from is $\{-1, 0, {+1}\}$. At the time of submission for Chan's Master's thesis, the largest number of eigenvalues successfully computed (using a machine with 32 GB of memory) were $32,767$, $17,710$, and $18,559$ for the Mandelbrot, Fibonacci-Mandelbrot, and Narayana-Mandelbrot matrices, respectively. This makes the \nth{16} Mandelbrot matrix the ``largest" Bohemian matrix that we have solved at the time we write this paper.


These new constructions led Chan and Corless to a new kind of companion matrix for polynomials of the form $c(z) = z a(z) b(z) + c_0$.  A first step towards this was first proved using the Schur complement in \cite{chan2017new}. 
Knuth then suggested that Chan and Corless look at the Euclid polynomials~\cite{chan2017minimal}, based on the Euclid numbers. 
It was the success of this construction that led to the realization that this construction is general, and gives a genuinely new kind of companion matrix.
Similar to the previous three families of matrices, the Euclid matrices are also upper Hessenberg and Bohemian, as the entries are comprised from the set $\{-1, 0, +1\}$. In addition, an interesting property of these companion matrices is that their inverses are also Bohemian with the same population, a property which we call ``the matrix family having \emph{rhapsody}~\cite{chan2017constructing}.''

As an extension of this generalization, Chan et al.~\cite{chan2017constructing} showed how to construct linearizations of matrix polynomials, particularly of the form $z\mathbf{a}(z)\mathbf{d}_0 + \mathbf{c}_0$, $\mathbf{a}(z)\mathbf{b}(z)$, $\mathbf{a}(z) +\mathbf{b}(z)$ (when $\mathrm{deg}(\mathbf{b}(z)) < \mathrm{deg}(\mathbf{a}(z))$, and $z\mathbf{a}(z)\mathbf{d}_0\mathbf{b}(z) + \mathbf{c}_0$, using a similar construction.

\section{Notation}
In what follows, we present some results on upper Hessenberg Bohemian matrices of the form
\begin{equation}
    \mathbf{H}_n =
    \renewcommand{\arraystretch}{1.3}
    \begin{bmatrix}
        h_{1,1} & h_{1,2}    & h_{1,3}     & \cdots & h_{1,n}\\
        s   & h_{2,2}    & h_{2,3}    & \cdots & h_{2,n}\\
        0     & s    & h_{3,3}    & \cdots & h_{3,n}\\
        \vdots & \ddots & \ddots & \ddots & \vdots\\
        0      & \cdots & 0       & s    & h_{n,n}
    \end{bmatrix}
\end{equation}
with $s = \exp(i\theta_k)$, usually $s  \in \{-1, {+1}\}$ (we do not allow zero subdiagonal entries, because that reduces the problem to smaller ones) and $h_{i,j} \in \{-1, 0, {+1}\}$ for $1 \le i \le j \le n$. We denote the characteristic polynomial $Q_n(z) \equiv \det (z \mathbf{I} - \mathbf{H}_n)$.

\begin{definition}
    The set of all $n \times n$ Bohemian upper Hessenberg matrices with upper triangle population $P$ and subdiagonal population from a discrete set of roots of unity, say $s\in \{e^{i\theta_{k}}\}$ where $\{\theta_{k}\}$ is some finite set of angles, is called $\mathcal{H}_{\{\theta_{k}\}}^{n\times n}(P)$. In particular, $\mathcal{H}_{\{0\}}^{n \times n}(P)$ is the set of all $n \times n$ Bohemian upper Hessenberg matrices with upper triangle entries from $P$ and subdiagonal entries equal to $1$ and $\mathcal{H}_{\{\pi\}}^{n \times n}(P)$ is when the subdiagonals entries are $-1$.
\end{definition}

It will often be true that the average value of a population will be zero.  In that case, matrices with trace zero will be common.  It is a useful oversimplification to look in that case at matrices whose diagonal is exactly zero. 

\begin{definition}
    For a population $P$ such that $0 \in P$, let $\ZH{\theta_k}{n}{P}$ be the subset of $\HH{\theta_k}{n}{P}$ where the main diagonal entries are fixed at 0.
\end{definition}

\section{Results of Experiments}
The methods used for computing the characteristic polynomials and counting the number of eigenvalues presented in Tables~\ref{tab:properties_UHTfulldiagonal}--\ref{tab:nilpotents} in this section will be discussed in detail in a forthcoming paper.  Many of the smaller-dimension computations were done directly in \textsc{Maple} 2017; for instance, computation of the characteristic polynomials of all two million or so matrices in \HH{0}{5}{\{0, {+1}\}} took about six hours on a Surface Pro. The greater number of higher-dimension matrices, or matrices with larger populations, required special techniques and larger \& faster machines.  Eigenvalue computations were also done in \textsc{Matlab} and in Python.  The computed characteristic polynomials are available through the Characteristic Polynomial Database~\cite{CPDB}.


\begin{table}[ht]
    \centering
    \begin{tabular}{|c|c|c|c|c|}
        \hline
        $n$ & \#matrices & \#cpolys &  \#neutral polys & \#neutral matrices \\
        \hline
         2 & 27 &  16&  2 & 4 \\
         \hline
         3 & 729  & 166 & 3 & 24 \\
         \hline
         4 & 59,049 & 3,317 & 7 & 332 \\
         \hline
         5 & 14,348,907 & 133,255 &  11 & 9,909 \\
         \hline
         6 & 10,460,353,203 & 10,872,459 & 25 & 696,083 \\
         \hline
    \end{tabular}
    \caption{Some properties of matrices in \HH{0}{n}{\{-1, 0, {+1}\}}. The \#matrices column reports the number of distinct matrices at each dimension.  The \#cpolys column reports the number of distinct characteristic polynomials at each dimension.  The \#neutral polys reports the number of characteristic polynomials where all roots have zero real part.  The \#neutral matrices column reports the number of matrices where all eigenvalues have zero real part.}
    \label{tab:properties_UHTfulldiagonal}
\end{table}
%
%
%
\begin{table}[ht]
    \centering
    \begin{tabular}{|c|c|c|c|c|}
        \hline
        $n$ & \#matrices & \#cpolys &  \#neutral polys & \#neutral matrices \\
        \hline
         2 & 3 & 3 & 2 & 2\\
         \hline
         3 & 27 & 15 & 3 & 6 \\
         \hline
         4 & 729 & 140 &  7 & 66 \\
         \hline
         5 & 59,049 & 2,297 & 11 & 1,069 \\
         \hline
         6 & 14,348,907 & 67,628 & 25 & 45,375 \\
         \hline
         7 & 10,460,353,203 & 3,606,225 & 45 & 4,105,977\\
         \hline
    \end{tabular}
    \caption{Some properties of matrices in \ZH{0}{n}{\{-1, 0, {+1}\}}. The \#matrices column reports the number of distinct matrices at each dimension.  The \#cpolys column reports the number of distinct characteristic polynomials at each dimension. The \#neutral polys reports the number of characteristic polynomials where all roots have zero real part.  The \#neutral matrices column reports the number of matrices where all eigenvalues have zero real part.}
    \label{tab:properties_UHT(3)}
\end{table}

\begin{table}[ht]
    \centering
    \begin{tabular}{|c|c|c|c|c|c|c|}
        \hline
         $n$ & multiplicity 1 & ${2}$ & ${3}$ & ${4}$ & ${5}$ & ${6}$\\
         \hline
         2 & 5 & 1 & & & & \\
         \hline
         3 & 35 & 0 & 1 & & & \\
         \hline
         4 & 431 & 5 & 0 & 1 & & \\
         \hline
         5 & 9,497 & 9 & 3 & 0 & 1 & \\
         \hline
         6 & 363,143 & 51 & 5 & 1 & 0 & 1 \\
         \hline
    \end{tabular}
    \caption{Number of distinct eigenvalues of various multiplicities of matrices in \ZH{0}{n}{\{-1, 0, {+1}\}}. Most eigenvalues are simple.  It turns out that every multiple eigenvalue also occurs as a simple eigenvalue for some other matrix.  The only $n$-multiple eigenvalue of the class of $n$ by $n$ matrices is, of course, $\lambda = 0$.}
    \label{tab:num_eigs(3)}
\end{table}

\begin{table}[h]
    \centering
    \begin{tabular}{|c|c|c|c|c|c|}
        \hline
        $n$ & \#matrices & \#cpolys & \#distinct real $\lambda$ & \#neutrals polys & \#neutral matrices \\
        \hline
         2 & 8 & 6 & 6 & 1 & 1\\
         \hline
         3 & 64 & 28 & 25 & 1 & 1\\
         \hline
         4 & 1,024 & 197 & 219 & 1 & 1\\
         \hline
         5 & 32,768 & 2,235 & 3,264 &  1 & 1\\
         \hline
         6 & 2,097,152 & 39,768 & 75,045 & 1 & 1\\
         \hline
         7 & 268,435,456 & 1,140,848 & 2,694,199 & 1 & 1\\
         \hline
    \end{tabular}
    \caption{Some properties of matrices in \HH{0}{n}{\{0, {+1}\}}. The \#matrices column reports the number of distinct matrices at each dimension.  The \#cpolys column reports the number of distinct characteristic polynomials at each dimension.
    The \#distinct real $\lambda$ column reports the number of distinct real eigenvalues in \HH{0}{n}{\{0, {+1}\}}.
    The \#neutral polys reports the number of characteristic polynomials where all roots have zero real part (here only $z^n$). We conjecture that this is always so (and that there is only one matrix for that neutral polynomial).  The \#neutral matrices column reports the number of matrices where all eigenvalues have zero real part.}
    \label{tab:properties_UHT(0-1 poly)}
\end{table}

\begin{table}[h]
    \centering
    \begin{tabular}{|c|c|c|c|c|c|c|}
        \hline
         $n$ & multiplicity 1 & ${2}$ & ${3}$ & ${4}$ & ${5}$ & ${6}$ \\
         \hline
         2 & 6 & 2 & & & & \\
         \hline
         3 & 43 & 2 & 2 & & & \\
         \hline
         4 & 413 & 6 & 2 & 2 & & \\
         \hline
         5 & 6,920 & 6 & 3 & 2 & 2 & \\
         \hline
         6 & 166,005 & 45 & 6 & 2 & 2 & 2 \\
         \hline
    \end{tabular}
    \caption{Number of distinct eigenvalues of various multiplicities matrices in $\HH{0}{n}{\{0, {+1}\}}$.  Note that in this class of matrices,  diagonal entries of the matrix need not be zero.}
    \label{tab:num_eigs(zero-one)}
\end{table}

\begin{table}[h]
    \centering
    \begin{tabular}{|c|c|c|c|c|c|c|}
        \hline
        $n$ & \#matrices & \#cpolys & \#stables & \#neutral polys & \#neutral matrices & \#distinct real $\lambda$ \\
        \hline
         2 & 8 & 6 & 1 & 1 & 2 & 5\\
         \hline
         3 & 64 & 32 & 3 & 0 & 0 & 29 \\
         \hline
         4 & 1,024 & 289 & 14 & 1& 6 & 233 \\
         \hline
         5 & 32,768 & 4,958 & 93 & 0 & 0 & 	7,363\\
         \hline
         6 & 2,097,152 & 162,059 & 992 & 2 & 430 & 299,477\\
         \hline
         7 & 268,435,456 & 10,318,948 & & 0 & 0 &\\
         \hline
    \end{tabular}
    \caption{Some properties of matrices from $\HH{0}{n}{\{-1, {+1}\}}$. The column \#stables reports the number of characteristic polynomials with all roots in the left half plane; the corresponding number of \textsl{matrices} is $1$, $4$, $28$, $424$, and $11,613$. Other columns are as in previous tables. Blank table entries represent unknowns.}
    \label{tab:properties_UHT(Bernoulli)}
\end{table}
\begin{table}[h]
    \centering
    \begin{tabular}{|c|c|c|c|c|c|c|}
        \hline
         $n$ & multiplicity 1 & ${2}$ & ${3}$ & ${4}$ & ${5}$ & ${6}$ \\
         \hline
         2 & 9 & 1 & & & & \\
         \hline
         3 & 65 & 0 & 0 & & & \\
         \hline
         4 & 689 & 5 & 0 & 0 & & \\
         \hline
         5 & 20,565 & 3 & 0 & 0 & 0 & \\
         \hline
         6 & 887,539 & 59 & 9 & 1 &  1  & 1 \\
         \hline
    \end{tabular}
    \caption{Number of distinct eigenvalues of various multiplicities of matrices from $\HH{0}{n}{\{-1, {+1}\}}$. The diagonal entries are not zero.}
    \label{tab:num_eigs(Bernoulli)}
\end{table}

Other questions than those answered in these tables can be asked of this data.  For instance, one might be interested in the proportion of singular matrices.  By asking which characteristic polynomials have zero constant coefficient, and counting the number of matrices that have that characteristic polynomial, one can answer such questions. In the case of six by six upper Hessenberg matrices with population $\{-1, {+1}\}$, there are $383,680$ singular matrices, or about $18.3\%$.  Recall that for ``random'' six by six matrices, where the entries are chosen perhaps uniformly over some real interval, the probability of singularity is \textsl{zero} because such matrices come from a set of measure zero.  Yet in applications, the probability of singular matrices is often \textsl{nonzero} because of structure.  By looking at Bohemian matrices, we get some idea of the  influence of structure for finite dimensions~$n$.

\section{Upper Hessenberg Matrices}
We can make sense of some of those experiments by theoretical results and proofs.  We begin with a recurrence relation for the characteristic polynomial $Q_n(z) = \det (z \mathbf{I} - \mathbf{H}_n)$ for $\mathbf{H}_n \in \HH{\theta_k}{n}{P}$ where $s = \exp(i\theta_k)$. Later we will specialize the population $P$ to contain only zero and numbers of unit magnitude, usually $\{-1,0,+1\}$.
\begin{theorem}
    \label{thm:charPolyRec1}
    \begin{equation}
        \label{eqn:thm1}
        Q_n(z) = zQ_{n-1}(z) - \sum_{k=1}^{n} s^{k-1} h_{n-k+1,n} Q_{n-k}(z)
    \end{equation}
    with the convention that $Q_0(z) = 1$ ($\mathbf{H}_0 = [\,]$, the empty matrix).
\end{theorem}
\begin{proof}
    We begin by proving the following equality:
    \begin{equation}
        \label{eq:thm1induction}
        \det
        \left[
        \begin{array}{cccc|c}
            \multicolumn{4}{c|}{\multirow{3}{*}{$z \mathbf{I} - \mathbf{H}_{i-1}$}} & -h_{1,n}\\
            \multicolumn{4}{c|}{} & \vdots\\
            \multicolumn{4}{c|}{} & -h_{i-1,n}\\ \hline
            0 & \cdots & 0 & -s & -h_{i,n}
        \end{array}
        \right] 
        = - \sum_{k=1}^i s^{k-1} h_{i-k+1,n} Q_{i-k}(z) 
    \end{equation}
    for $1 \le i \le n$.
    
    When $i = 1$ the left side of equation~\eqref{eq:thm1induction} reduces to $\det \begin{bmatrix} -h_{1,n} \end{bmatrix} = -h_{1,n}$, and the right side reduces to $-\sum_{k=1}^1 s^{k-1} h_{1-k+1,n} Q_{1-k}(z) = -h_{1,n}$.
    
    Assume inductively that
    \begin{equation}
        \det
        \left[
        \begin{array}{cccc|c}
            \multicolumn{4}{c|}{\multirow{3}{*}{$z \mathbf{I} - \mathbf{H}_{j-1}$}} & -h_{1,n}\\
            \multicolumn{4}{c|}{} & \vdots\\
            \multicolumn{4}{c|}{} & -h_{j-1,n}\\ \hline
            0 & \cdots & 0 & -s & -h_{j,n}
        \end{array}
        \right] 
        = - \sum_{k=1}^j s^{k-1} h_{j-k+1,n} Q_{j-k}(z)
    \end{equation}
    for $i = j-1$. Then
    \begin{align}
        \det
        \left[
        \begin{array}{cccc|c}
            \multicolumn{4}{c|}{\multirow{3}{*}{$z \mathbf{I} - \mathbf{H}_{j}$}} & -h_{1,n}\\
            \multicolumn{4}{c|}{} & \vdots\\
            \multicolumn{4}{c|}{} & -h_{j,n}\\ \hline
            0 & \cdots & 0 & -s & -h_{j+1,n}
        \end{array}
        \right]
        &= -h_{j+1,n} \det (z \mathbf{I} - \mathbf{H}_j) + s \det
        \left[
        \begin{array}{cccc|c}
            \multicolumn{4}{c|}{\multirow{3}{*}{$z \mathbf{I} - \mathbf{H}_{j-1}$}} & -h_{1,n}\\
            \multicolumn{4}{c|}{} & \vdots\\
            \multicolumn{4}{c|}{} & -h_{j-1,n}\\ \hline
            0 & \cdots & 0 & -s & -h_{j,n}
        \end{array}
        \right] \notag\\
        &= -h_{j+1,n}Q_j(z) + s \left ( - \sum_{k=1}^j s^{k-1} h_{j-k+1,n} Q_{j-k}(z) \right )\\
        &= -h_{j+1,n}Q_j(z) - \sum_{k=1}^j s^k h_{j-k+1,n} Q_{j-k}(z)\\
        &= - \sum_{k=0}^j s^k h_{j-k+1,n} Q_{j-k}(z)\\
        &= - \sum_{k=1}^{j+1} s^{k-1} h_{(j+1)-k+1,n} Q_{(j+1)-k}(z) \>.
    \end{align}
Next we prove the theorem. Performing Laplace expansion on the last row of $z \mathbf{I} - \mathbf{H}_n$ we get
    \begin{align}
        Q_n(z) &= \det
        \left[
        \begin{array}{cccc|c}
            \multicolumn{4}{c|}{\multirow{3}{*}{$z \mathbf{I} - \mathbf{H}_{n-1}$}} & -h_{1,n}\\
            \multicolumn{4}{c|}{} & \vdots\\
            \multicolumn{4}{c|}{} & -h_{n-1,n}\\ \hline
            0 & \cdots & 0 & -s & z-h_{n,n}
        \end{array}
        \right] \\
        &= (z-h_{n,n}) \det(z \mathbf{I} - \mathbf{H}_{n-1}) + s \det
        \left[
        \begin{array}{cccc|c}
            \multicolumn{4}{c|}{\multirow{3}{*}{$z \mathbf{I} - \mathbf{H}_{n-2}$}} & -h_{1,n}\\
            \multicolumn{4}{c|}{} & \vdots\\
            \multicolumn{4}{c|}{} & -h_{n-2,n}\\ \hline
            0 & \cdots & 0 & -s & -h_{n-1,n}
        \end{array}
        \right] \\
        &= z Q_{n-1}(z) - h_{n,n}Q_{n-1}(z) + s \left ( - \sum_{k=1}^{n-1} s^{k-1} h_{n-1-k+1,n} Q_{n-1-k}(z) \right)\\
        &= z Q_{n-1}(z) - h_{n,n} Q_{n-1}(z) - \sum_{k=1}^{n-1} s^k h_{n-k,n} Q_{n-1-k}(z)\\
        &= z Q_{n-1}(z) - \sum_{k=0}^{n-1} s^k h_{n-k,n} Q_{n-1-k}(z)\\
        &= zQ_{n-1}(z) - \sum_{k=1}^n s^{k-1} h_{n-k+1,n} Q_{n-k}(z) \>.
    \end{align}
\end{proof}

\begin{theorem}
    \label{thm:charPolyRec2}
    Expanding $Q_n(z)$ as
    \begin{equation}
        Q_n(z) = q_{n,n} z^n + q_{n,n-1} z^{n-1} + \cdots + q_{n,0},
    \end{equation}
    we can express the coefficients recursively by
    \begin{subequations}
    \label{eqn:thm2_eqn}
    \begin{align}
        q_{n,n} &= 1,\\
        q_{n,j} &= q_{n-1,j-1} - \sum_{k=1}^{n-j} s^{k-1} h_{n-k+1,n} q_{n-k,j} \quad\text{for}\quad 1 \le j \le n-1,\\
        q_{n,0} &= -\sum_{k=1}^n s^{k-1} h_{n-k+1,n} q_{n-k,0} \quad\text{for}\quad n>0,\quad\text{and}\\
        q_{0,0} &= 1\>.
    \end{align}
    \end{subequations}
\end{theorem}
\begin{proof}
    By Theorem~\ref{thm:charPolyRec1}
    \begin{equation}
        Q_n(z) = zQ_{n-1}(z) - \sum_{k=1}^n s^{k-1} h_{n-k+1,n} Q_{n-k}(z) \>.
    \end{equation}
    The first term can be written
    \begin{align}
        zQ_{n-1}(z) &= z \left [z^{n-1} + q_{n-1,n-2} z^{n-2} + \cdots + q_{n-1, 0} \right ]\\
        &= z \left [z^{n-1} + \sum_{j=0}^{n-2} q_{n-1,j} z^j \right ]\\
        &= z^n + \sum_{j=0}^{n-2} q_{n-1,j} z^{j+1}\\
        &= z^n + \sum_{j=1}^{n-1} q_{n-1,j-1} z^{j}
    \end{align}
    and the second term
    \begin{align}
        s^{k-1} h_{n-k+1,n} Q_{n-k}(z) &= s^{k-1} h_{n-k+1,n} \left [q_{n-k,n-k} z^{n-k} + q_{n-k,n-k-1}z^{n-k-1} + \cdots + q_{n-k,0}\right ] \notag\\
        &= s^{k-1} h_{n-k+1,n} \sum_{j=0}^{n-k} q_{n-k,j} z^j \>.
    \end{align}
    Therefore,
    \begin{align*}
        Q_n(z) &= z^n + \sum_{j=1}^{n-1} q_{n-1,j-1} z^j - \sum_{k=1}^n s^{k-1} h_{n-k+1,n} \sum_{j=0}^{n-k} q_{n-k,j} z^j\\
        &= z^n + \sum_{j=1}^{n-1} q_{n-1,j-1} z^j - \sum_{j=0}^{n-1} \left ( \sum_{k=1}^{n-j} s^{k-1} h_{n-k+1,n} q_{n-k,j} \right ) z^j\\
        &= z^n + \sum_{j=1}^{n-1} \left ( q_{n-1,j-1} - \sum_{k=1}^{n-j} s^{k-1} h_{n-k+1,n} q_{n-k,j} \right ) z^j - \sum_{k=1}^n s^{k-1} h_{n-k+1,n} q_{n-k,0} \>.
    \end{align*}
\end{proof}

\begin{proposition}
All matrices in $\HH{\theta_k}{n}{P}$ are non-derogatory\footnote{A non-derogatory matrix is a matrix for which its characteristic polynomial and minimal polynomial coincide (up to a factor of $\pm 1$)}.
\end{proposition}
\begin{proof}
Let $\mathbf{H} \in \HH{\theta_k}{n}{P}$. Because $\mathbf{H}$ is upper Hessenberg
\begin{equation}
    \mathbf{H}_{i,j}^k = 
    \begin{cases}
    f_{i,j,k} &\quad\text{for}\quad i < j+k\\
    s^k &\quad\text{for}\quad i = j + k\\
    0 &\quad\text{for}\quad i > j+k\\
    \end{cases}
\end{equation}
for $0 \le k \le n-1$ where $f_{i,j,k}$ are some functions of the entries of $\mathbf{H}$. Let
\begin{equation}
    \mathbf{A} = r(\mathbf{H}) = \sum_{k=0}^{n-1} c_{k} \mathbf{H}^k  = \mathbf{0}\> .
\end{equation}
We find $\mathbf{A}_{n,1} = s^{n-1} c_{n-1} = 0$ and therefore $c_{n-1} = 0$. Continuing recursively for $k$ from $n-2$ to 1 we find $\mathbf{A}_{k+j,j} = s^k c_k = 0$ for $1 \le j \le n-k$ and therefore $c_k = 0$ (since $c_j = 0$ for $j > k$) for $1 \le k \le n-1$. We have $\mathbf{A} = c_0 \mathbf{H}^0 = \mathbf{0}$ and hence $c_0 = 0$. Thus, no non-zero polynomial of degree less than $n$ exists that satisfies $r(\mathbf{H}) = \mathbf{0}$. Therefore, the minimal degree non-zero polynomial that satisfies $r(\mathbf{H}) = \mathbf{0}$ is the characteristic polynomial of $\mathbf{H}$.
\end{proof}

\begin{definition}
    The \textit{characteristic height} of a matrix is the height of its 
    characteristic polynomial.
\end{definition}
\begin{remark}
    The height of a polynomial is in fact a norm (the infinity norm of the vector of coefficients).
\end{remark}

\begin{proposition}
    \label{prop:negativeheight}
    For any matrix $\mathbf{A}$, $-\mathbf{A}$ has the same characteristic height as $\mathbf{A}$.
\end{proposition}
\begin{proposition}
    \label{prop:maxheight}
    The maximal characteristic height of $\mathbf{H}_n \in \mathcal{H}_{\{0,\pi\}}^{n \times n}(\{-1, 0, {+1}\})$ occurs when
    $s^{k-1} h_{i,i+k-1} = -1$ for $1 \le i \le n-k+1$ and $1 \le k \le n$.
\end{proposition}
\begin{proof}
    Since $s \in \{-1, {+1}\}$ and $h_{i,j} \in \{-1, 0, {+1}\}$, $s^{k-1} h_{i,i+k-1} \in \{-1, 0, {+1}\}$ {and hence $\max |s^{k-1} h_{i,i+k-1}| = 1$.} Let $s^{k-1} h_{i,i+k-1} = -1$. By Theorem~\ref{thm:charPolyRec2}
    \begin{align}
        q_{n,0} &= -\sum_{k=1}^n s^{k-1} h_{n-k+1,n} q_{n-k,0}\\
                &= \sum_{k=1}^n q_{n-k,0} \label{eq:prop1_3_MaxHeight}
    \end{align}
    and
    \begin{align}
        q_{n,j} &= q_{n-1,j-1} - \sum_{k=1}^{n-j} s^{k-1} h_{n-k+1,n} q_{n-k,j}\\
                &= q_{n-1,j-1} + \sum_{k=1}^{n-j} q_{n-k,j} \>. \label{eq:prop1_2_MaxHeight}
    \end{align}
    Since $q_{0,0} = 1$, and equations~\eqref{eq:prop1_3_MaxHeight} and \eqref{eq:prop1_2_MaxHeight} are independent of $s$ and $h_{i,j}$, all $q_{n,j}$ must be positive and the maximum characteristic height is attained.
\end{proof}
\begin{remark}
    When $s = 1$ ($\theta = 0$) and $h_{i,j} = -1$ for all $1 \le i \le j \le n$, $\mathbf{H}_n$ attains maximal characteristic height. By Proposition~\ref{prop:negativeheight}, $s = -1$ ($\theta = \pi$) and $h_{i,j} = 1$ will also attain maximal characteristic height. Both of these cases correspond to upper Hessenberg matrices with a Toeplitz structure as we explore in further detail in the paper~\cite{chan2018BUHT}.
\end{remark}

\begin{definition}
    $P$ is invariant under multiplication by a fixed unit $e^{i\theta}$ if $e^{i\theta}P = P$; that is, each entry of $P$, say $p$, is such that $e^{i\theta}p$ is also in $P$. For instance, $\{-1, 0, {+1}\}$ is invariant under multiplication by $-1$. Note that invariance with respect to $e^{i\theta}$ implies invariance with respect to $e^{-i\theta}$.
\end{definition}

\begin{theorem}
    Suppose $\mathbf{H}_n \in \HH{\theta_k}{n}{P}$ and $P$ is invariant under multiplication by each $e^{i\theta_{k}}$ and by $-e^{i\theta_{k}}$. Then $\mathbf{H}_n$ is similar to a matrix in $\HH{\pi}{n}{P}$, and similar to a matrix in $\HH{0}{n}{P}$.
\end{theorem}

\begin{proof}
    We use induction. The case $n=1$ is vacuously upper Hessenberg, though it is
    \begin{equation*}
        \begin{bmatrix}
            e^{i\theta_{k}}
        \end{bmatrix}
        \begin{bmatrix}
            h_{11}
        \end{bmatrix}
        \begin{bmatrix}
            e^{-i\theta_{k}}
        \end{bmatrix}
        =
        \begin{bmatrix}
            h_{11}
        \end{bmatrix}
        \in \HH{\theta_k}{1}{P} \>.
    \end{equation*}
    For $n > 1$, partition the matrix as
    \begin{equation*}
        \left[
        \begin{array}{cccc}
            h_{11} & h_{12} \cdots & h_{1n} \\
            s & \multicolumn{3}{c}{\multirow{3}{*}{$\mathbf{H}_{n-1}$}} \\
            & & & \\
            & & &
        \end{array}
        \right]
    \end{equation*}
    where $s = e^{i\theta_{k}}$ for some $\theta_{k}$. Then conjugate by
    \begin{multline*}
        \begin{bmatrix}
            1 & & \\
            & -e^{i\theta_{k}} & \\
            & & \mathbf{I}_{n-2}
        \end{bmatrix}
        \left[
        \begin{array}{cccc}
            h_{11} & h_{12} \cdots & h_{1n} \\
            s & \multicolumn{3}{c}{\multirow{3}{*}{$\mathbf{H}_{n-1}$}} \\
            & & & \\
            & & &
        \end{array}
        \right]
        \begin{bmatrix}
            1 & & \\
            & -e^{i\theta_{k}} & \\
            & & \mathbf{I}_{n-2}
        \end{bmatrix} ^{-1} \\
        =
        \left[
        \begin{array}{ccc}
            h_{11} & -e^{-i\theta_{k}}h_{12} & \cdots \\
            -1 & \multicolumn{2}{c}{\multirow{2}{*}{$\tilde{\mathbf{H}}_{n-1}$}} \\
            & &
        \end{array}
        \right] \>.
    \end{multline*}
    Clearly $\tilde{\mathbf{H}}_{n-1}$ is in $\HH{\theta_k}{n-1}{P}$. By induction the proof is complete.
\end{proof}

\begin{remark}
    For clarity, consider the case $n = 2$:
    \begin{equation}
         \mathbf{H} =
        \begin{bmatrix}
            a & b \\
            s & c
        \end{bmatrix} \>,
    \end{equation}
    where $a, b, c \in P$ and $s = e^{i\theta_{k}}$. Then, the following similarity transforms reduce the problem to one in \HH{0}{2}{P} and one in \HH{\pi}{2}{P}.
    \begin{align}
        \begin{bmatrix}
            1 & 0 \\
            0 & e^{-i\theta_{k}}
        \end{bmatrix}
        \mathbf{H}
        \begin{bmatrix}
            1 & 0 \\
            0 & e^{i\theta_{k}}
        \end{bmatrix}
        & =
        \begin{bmatrix}
            a & be^{i\theta_{k}} \\
            1 & c
        \end{bmatrix} \\
        \begin{bmatrix}
            1 & 0 \\
            0 & -e^{-i\theta_{k}}
        \end{bmatrix}
        \mathbf{H}
        \begin{bmatrix}
            1 & 0 \\
            0 & -e^{i\theta_{k}}
        \end{bmatrix}
        &= 
        \begin{bmatrix}
            a & -be^{i\theta_{k}} \\
            -1 & c
        \end{bmatrix} \>.
    \end{align}
\end{remark}


\section{Upper Hessenberg Toeplitz Matrices}
Proposition~\ref{prop:maxheight} gives  matrices in 
\HH{0,\pi}{n}{\{-1, 0, {+1}\}}
with maximal characteristic height\footnote{We did not report the numbers of such matrices and polynomials that we found in our ``results'' section.}. We noticed that they are Toeplitz matrices. This motivated our interest in upper Hessenberg Toeplitz matrices. We summarize some of the results of~\cite{chan2018BUHT} here.

\begin{definition}
    Consider matrices $\mathbf{H}_n$ where
    $h_{i,i+k-1} = t_k$ for $1 \le i \le n-k+1$, $1 \le k \le n$ and $s = 1$. We denote these matrices by $\mathbf{M}_n$ and they have a Toeplitz structure.
\end{definition}

\begin{remark}
    \label{thm:maxheight_UHT}
    The characteristic height of $\mathbf{M}_n$ is maximal when $t_k = -1$ for $1 \le k \le n$.
\end{remark}

\begin{proposition}
    Let $F \subset \mathbb{R}$ be a closed and bounded set with $a = \min{F}$, $b = \max{F}$ and $\#F \ge 2$. Let $\mathbf{M}_n$ be upper Hessenberg Toeplitz with $t_k \in F$. If $|a| \ge |b|$, $\mathbf{M}_n$ attains maximal characteristic height for $t_k = a$ for all $1 \le k \le n$. If $|b| \ge |a|$, $\mathbf{M}_n$ attains maximal characteristic height for $t_k = a$ for $k$ even, and $t_k = b$ for $k$ odd.
\end{proposition}

\begin{proposition}
    The maximum characteristic height grows at least exponentially in $n$.
\end{proposition}

\begin{conjecture}
    The maximum characteristic height approaches $C (1 + \varphi)^n$ as $n \to \infty$ for some constant $C$ where $\varphi$ is the golden ratio.
\end{conjecture}
\begin{remark}
    This limit is illustrated in Figure~\ref{fig:maxheightratiolog}, motivating this conjecture.
\end{remark}
\begin{figure}[h]
    \centering
    \includegraphics[width=\textwidth]{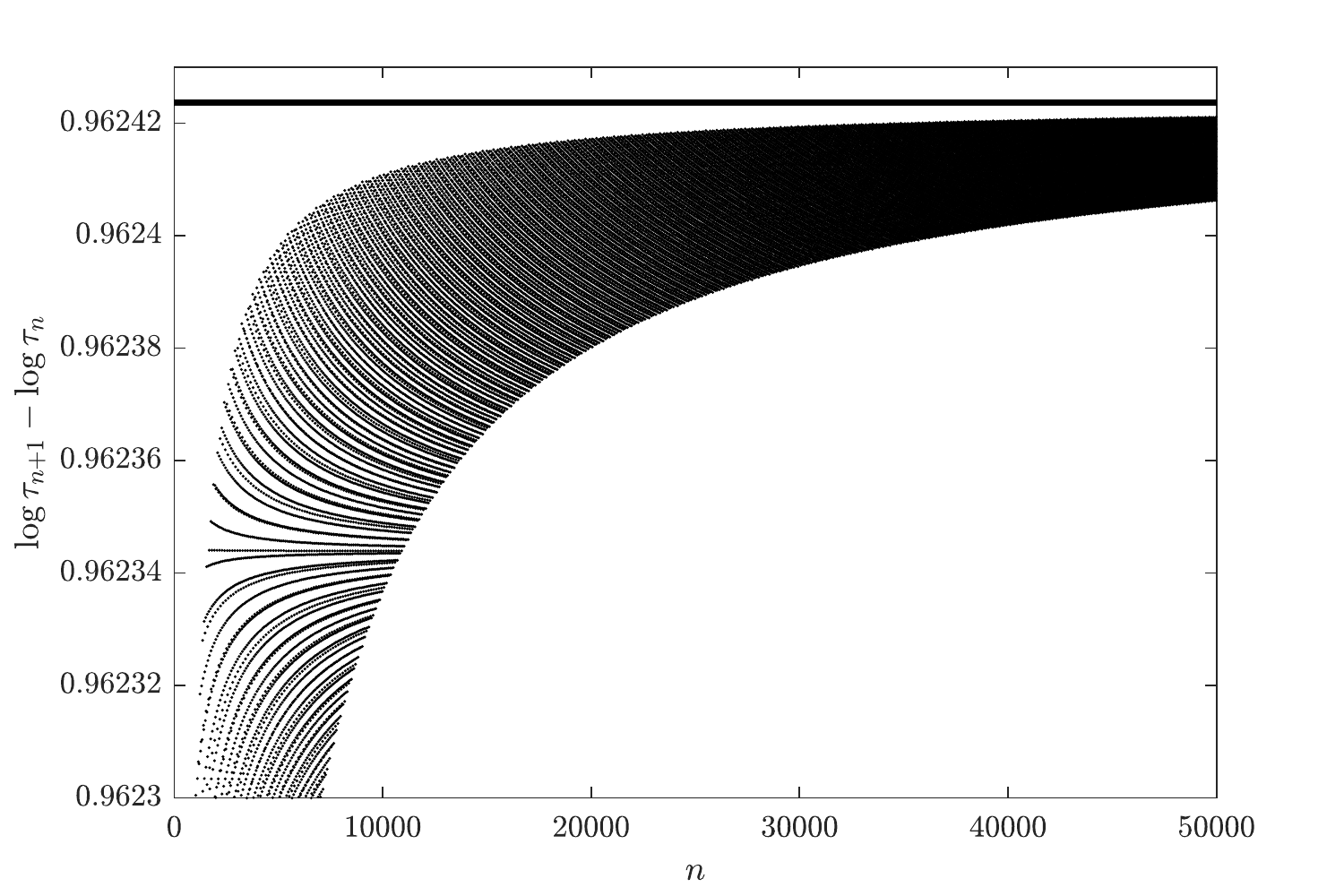}
    \caption{The points are $\log{\tau_{n+1}} - \log{\tau_n}$ for $n$ from 0 to 50000 where $\tau_n$ is the maximal characteristic height of $\mathbf{M}_n$ (i.e. when $t_k = 1$, for example). The solid line is $\log(1 + \varphi)$ where $\varphi$ is the golden ratio.}
    \label{fig:maxheightratiolog}
\end{figure}

\section{Zero Diagonal Upper Hessenberg Matrices}


\begin{theorem}
\label{thm:zero_diag_UH}
Let $\mathbf{A}_n \in \ZH{0}{n}{P}$ for $P = \{0, w_1, \ldots, w_m\}$ for some fixed positive integer $m$ and each $|w_j| = 1$. If $\mathbf{A}_{n}$ is normal, i.e.~$\mathbf{A}_{n}^{*}\mathbf{A}_{n} = \mathbf{A}_{n}\mathbf{A}_{n}^{*}$, then for $n \geq 3$, $\mathbf{A}_{n}$ is symmetric, $w_{j}$-skew symmetric for some fixed $1 \leq j \leq m$ or $w_{j}$-skew circulant. These $2m$ matrices ($m$ symmetric/$w_j$-skew symmetric, and $m$ $w_j$-skew circulant matrices) are the only normal matrices in $\ZH{0}{n}{P}$. (For $n=1$, this is only $\left[ 0 \right]$; for $n=2$, the symmetric and circulant cases coalesce, so that there are only $m$ such matrices.)
\end{theorem}
\begin{proof}
To prove this theorem, we establish a sequence of lemmas. First, we partition $\mathbf{A}_{n}$. Put
\begin{equation}
    \mathbf{A}_{n} =
    \begin{bmatrix}
        0 & \mathbf{T}^{*} \\
        \mathbf{e} & \mathbf{A}_{n-1}
    \end{bmatrix}
\end{equation}
where
\begin{equation}
    \mathbf{e}^{*} =
    \begin{bmatrix}
        1 & 0 & \cdots & 0
    \end{bmatrix}
\end{equation}
and
\begin{equation}
    \mathbf{T}^{*} =
    \begin{bmatrix}
        t_{12} & t_{13} & \cdots & t_{1n}
    \end{bmatrix}
    =
    \begin{bmatrix}
        t_{21}^{*} & t_{31}^{*} \cdots t_{n1}^{*}
    \end{bmatrix}^{*} \>.
\end{equation}
Then the conditions of normality are
\begin{equation}
    \mathbf{A}_{n}\mathbf{A}_{n}^{*} =
    \begin{bmatrix}
        \mathbf{T}^{*}\mathbf{T} & \mathbf{T}^{*}\mathbf{A}_{n-1}^{*} \\
        \mathbf{A}_{n-1} & \mathbf{e}\mathbf{e}^{*} + \mathbf{A}_{n-1}\mathbf{A}_{n-1}^{*}
    \end{bmatrix}
\end{equation}
must equal
\begin{equation}
    \mathbf{A}_{n}^{*}\mathbf{A}_{n} =
    \begin{bmatrix}
        1 & \mathbf{e}^{*}\mathbf{A}_{n-1} \\
        \mathbf{A}_{n-1}^{*}\mathbf{e} & \mathbf{T}\mathbf{T}^{*} + \mathbf{A}_{n-1}^{*}\mathbf{A}_{n-1}
    \end{bmatrix} \>.
\end{equation}
\end{proof}

\begin{lemma}
\label{lemma:nonzero_Mn}
The first row of $\mathbf{A}_{n}$ contains exactly one nonzero element, say $\tau$ in position $j$ $(2 \leq j \leq n)$.
\end{lemma}

\begin{proof}
\begin{equation}
    \mathbf{T}^{*}\mathbf{T} = \sum_{j=2}^{n} \left|t_{ij}\right|^{2} = 1
\end{equation}
from the upper left corner. Since each nonzero element of $P$ has magnitude $1$, exactly one entry must be nonzero.
\end{proof}

\begin{lemma}
If $\mathbf{A}_{n-1}$ is normal then $\mathbf{T} = \tau\mathbf{e}$ and $\mathbf{A}_{n}$ is $\tau$-skew symmetric.
\end{lemma}

\begin{proof}
If $\mathbf{A}_{n-1}$ is normal, then $\mathbf{T}\mathbf{T}^{*} + \mathbf{A}_{n-1}^{*}\mathbf{A}_{n-1}$ being equal to $\mathbf{e}\mathbf{e}^{*} + \mathbf{A}_{n-1}\mathbf{A}_{n-1}^{*}$ implies $\mathbf{T}\mathbf{T}^{*} = \mathbf{e}\mathbf{e}^{*}$ so that $\mathbf{T}^{*} = \begin{bmatrix}\tau^{*} & 0 & \cdots & 0\end{bmatrix}$ for some $\tau$ with $\left|\tau\right| = 1$. Then 
\begin{equation}
\mathbf{T}^{*}\mathbf{A}_{n-1}^{*} = \mathbf{e}^{*}\mathbf{A}_{n-1} \Rightarrow \tau^{*}\begin{bmatrix} 1 & 0 & \cdots & 0 \end{bmatrix}\mathbf{A}_{n-1}^{*} = \mathbf{e}^{*}\mathbf{A}_{n-1}^{*}
\end{equation}
and this says $\tau^{*}$ times the first row of $\mathbf{A}_{n-1}^{*}$ is the first row of $\mathbf{A}_{n-1}$.

But the first row of $\mathbf{A}_{n-1}^{*}$ is $\begin{bmatrix}0 & 1 & 0 & \cdots & 0\end{bmatrix}$ because $\mathbf{A}_{n-1}$ is upper Hessenberg with zero diagonal. Thus the first row of $\mathbf{A}_{n-1}$ is $\begin{bmatrix}0 & \tau^{*} & 0 & \cdots & 0\end{bmatrix}$. Thus
\begin{equation}
    \mathbf{A}_{n} =
    \left[
        \begin{array}{c|c|c}
            0 & \tau^{*} & \\
            \hline
            1 & 0 & \begin{array}{ccc}\tau^{*} & &\end{array} \\
            \hline
            & \begin{array}{c}1 \\ \\\end{array} & \mathbf{A}_{n-2}
        \end{array}
    \right]
    \quad \text{(remember $n \geq 3$)}
\end{equation}
and
\begin{equation}
    \mathbf{A}_{n-1}=
    \left[
        \begin{array}{cc}
            0 & \begin{array}{ccc} \tau^{*} & & \end{array} \\
            \begin{array}{c} 1 \\ \\ \end{array} & \mathbf{A}_{n-2}
        \end{array}
    \right]
\end{equation}
is normal. Because $\mathbf{A}_{n-1}$ is normal, and 
\begin{equation}
    \mathbf{A}_{n-1}^{*} = 
    \left[
        \begin{array}{ccc}
            0 & 1 & \\
            \tau & 0 & \begin{array}{ccc}1 & & \end{array} \\
            & \begin{array}{c}\tau \\ \\ \end{array} & \mathbf{A}_{n-2}^{*}
        \end{array}
    \right]
\end{equation}
we have $\mathbf{A}_{n-1}^{*}\mathbf{A}_{n-1} = \mathbf{A}_{n-1}\mathbf{A}_{n-1}^{*}$ or
\begin{multline*}
    \left[
        \begin{array}{ccc}
            0 & 1 & \\
            \tau & 0 & \begin{array}{ccc} 1 & & \end{array} \\
            & \begin{array}{c}1 \\ \\ \end{array} & \mathbf{A}_{n-2}
        \end{array}
    \right]
    \left[
        \begin{array}{ccc}
            0 & \tau^{*} & \\
            1 & 0 & \begin{array}{ccc}\tau^{*} & & \end{array} \\
            & \begin{array}{c}1 \\ \\ \end{array} & \mathbf{A}_{n-2}
        \end{array}
    \right] \\
    =
    \left[
        \begin{array}{ccl}
            1 & 0 & \tau^{*} \\
            0 & 2 & \mathbf{e}_{n-2}^{*}\mathbf{A}_{n-2} \\
            \begin{array}{c}\tau \\ \\ \end{array} & \tau\mathbf{A}_{n-2}^{+}\mathbf{e}_{n-2} & \mathbf{e}\mathbf{e}^{*} + \mathbf{A}_{n-2}^{*}\mathbf{A}_{n-2}
        \end{array}
    \right]
\end{multline*}
must equal
\begin{multline*}
    \left[
        \begin{array}{ccc}
            0 & \tau^{*} & \\
            1 & 0 & \begin{array}{ccc}\tau^{*} & & \end{array} \\
            & \begin{array}{c}1 \\ \\ \end{array} & \mathbf{A}_{n-2}
        \end{array}
    \right]
    \left[
        \begin{array}{ccc}
            0 & 1 & \\
            \tau & 0 & \begin{array}{ccc} 1 & & \end{array} \\
            & \begin{array}{c}1 \\ \\ \end{array} & \mathbf{A}_{n-2}
        \end{array}
    \right] \\
    =
    \left[
        \begin{array}{ccl}
            1 & 0 & \tau^{*} \\
            0 & 2 & \mathbf{e}_{n-2}^{*}\mathbf{A}_{n-2} \\
            \begin{array}{c}\tau \\ \\ \end{array} & \tau\mathbf{A}_{n-2}^{+}\mathbf{e}_{n-2} & \mathbf{e}\mathbf{e}^{*} + \mathbf{A}_{n-2}^{*}\mathbf{A}_{n-2}
        \end{array}
    \right] \>.
\end{multline*}
The lower left block gives $\mathbf{e}\mathbf{e}^{*} + \mathbf{A}_{n-2}\mathbf{A}_{n-2}^{*} = \mathbf{e}\mathbf{e}^{*} + \mathbf{A}_{n-2}^{*}\mathbf{A}_{n-2}$ so $\mathbf{A}_{n-2}$ must also be normal.

At this point, we see the outline of an induction:
\begin{equation}
    \mathbf{A}_{n} =
    \left[
        \begin{array}{c|c}
            0 & \begin{array}{ccc}\tau^{*} & & \end{array} \\
            \hline
            \begin{array}{c}1 \\ \\ \end{array} & \mathbf{A}_{n-1}
        \end{array}
    \right]
\end{equation}
being normal with $\mathbf{A}_{n-1}$ being normal implies that
\begin{equation}
    \mathbf{A}_{n-1} =
    \left[
        \begin{array}{c|c}
            0 & \begin{array}{ccc}\tau^{*} & & \end{array} \\
            \hline
            \begin{array}{c}1 \\ \\ \end{array} & \mathbf{A}_{n-2}
        \end{array}
    \right]
\end{equation}
where $\mathbf{A}_{n-2}$ is normal. Explicit computation of the $n=3$ case shows the induction terminates.
\end{proof}

We now consider the harder case where
\begin{equation}
    \mathbf{A}_{n} =
    \begin{bmatrix}
        0 & \mathbf{T}^{*} \\
        \mathbf{e}_{n-1} & \mathbf{A}_{n-1}
    \end{bmatrix}
\end{equation}
but where $\mathbf{A}_{n-1}$ is not itself normal. From Lemma \ref{lemma:nonzero_Mn} we know that $\mathbf{T}^{*}$ has only one nonzero element; call it $\tau^{*}$ as before. Then
\begin{equation}
    \mathbf{T}\mathbf{T}^{*} =
    \begin{bmatrix}
        0 & & & & & & \\
        & \ddots & & & & & \\
        & & 0 & & & & \\
        & & & 1 & & & \\
        & & & & 0 & & \\
        & & & & & \ddots & \\
        & & & & & & 0
    \end{bmatrix}
\end{equation}
while
\begin{equation}
    \mathbf{e}\mathbf{e}^{*} =
    \begin{bmatrix}
        1 & & & \\
        & 0 & & \\
        & & \ddots & \\
        & & & 0
    \end{bmatrix} \>,
\end{equation}
and we may assume that the $1$ in $\mathbf{T}\mathbf{T}^{*}$ does not occur in the first row and column (else we are in the previous case, and $\mathbf{A}_{n-1}$ will be normal). Here 
\begin{equation}
    \label{eq:Mnormal}
    \mathbf{A}_{n-1}\mathbf{A}_{n-1}^{*} - \mathbf{A}_{n-1}^{*}\mathbf{A}_{n-1} = \mathbf{T}\mathbf{T}^{*} - \mathbf{e}\mathbf{e}^{*} = 
    \begin{bmatrix}
        -1 & & & & & & & \\
        & 0 & & & & & & \\
        & & \ddots & & & & & \\
        & & & 0 & & & & \\
        & & & & 1 & & & \\
        & & & & & 0 & & \\
        & & & & & & \ddots & \\
        & & & & & & & 0
    \end{bmatrix}
\end{equation}
is the departure of $\mathbf{A}_{n-1}$ from normality. We will establish that in fact 
\begin{equation}
    \mathbf{T}^{*} =
    \begin{bmatrix}
        0 & 0 & 0 & \cdots & 0 & \tau^{*}
    \end{bmatrix}
\end{equation}
and that
\begin{equation}
    \mathbf{A}_{n-1} =
    \begin{bmatrix}
        0 & & & & \\
        1 & 0 & & & \\
        & 1 & 0 & & \\
        & & \ddots & \ddots & \\
        & & & 1 & 0
    \end{bmatrix} \>;
\end{equation}
that is, the nonzero element can only occur in the last place. Notice that the upper left corner of \ref{eq:Mnormal} is, if the top row of $\mathbf{A}_{n-1}$ is $\begin{bmatrix}0 & a_{1,2} & a_{1,3} & \cdots a_{1,n-1}\end{bmatrix}$,
\begin{equation}
    \sum_{j=2}^{n-1} \left|a_{1,j}\right|^{2} - 1 \>.
\end{equation}
Therefore, all $a_{1,j} = 0$ and the first row of $\mathbf{A}_{n-1}$ must be zero: i.e.
\begin{equation}
    \mathbf{A}_{n-1} =
    \begin{bmatrix}
        0 & 0 & 0 & \cdots & 0 \\
        1 & 0 & a_{3,3} & \cdots & a_{2, n-1} \\
        & 1 & 0 & \ddots & \vdots \\
        & & \ddots & \ddots & a_{n-2, n-1} \\
        & & & 1 & 0 
    \end{bmatrix}
\end{equation}
Then,
\begin{equation}
    \mathbf{A}_{n-1}\mathbf{T} = \mathbf{A}_{n-1}^{*}\mathbf{e} = 
    \begin{bmatrix}
        0 & 1 & & \\
        0 & 0 & \ddots \\
        \vdots & & \ddots & 1 \\
        0 & \cdots & \cdots & 0 
    \end{bmatrix}
    \begin{bmatrix}
        1 \\
        0 \\
        \vdots \\
        0
    \end{bmatrix}
    =
    \begin{bmatrix}
        0 \\
        0 \\
        \vdots \\
        0
    \end{bmatrix} \>.
\end{equation}
If
\begin{equation}
    \mathbf{T} =
    \begin{bmatrix}
        0 \\
        0 \\
        \vdots \\
        0 \\
        \tau \\
        0 \\
        \vdots \\
        0
    \end{bmatrix}\>,
\end{equation}
then
\begin{equation}
    \mathbf{A}_{n-1}\mathbf{T} =
    \begin{bmatrix}
        0 \\
        \tau a_{2,j} \\
        \vdots \\
        \tau a_{j-1, j} \\
        0 \\
        \tau \\
        0 \\
        \vdots \\
        0
    \end{bmatrix}\>,
\end{equation}
which is impossible unless $j=n$ (when the $\tau$ term is not present). Therefore,
\begin{equation}
    \mathbf{A}_{n-1} =
    \begin{bmatrix}
        0 & 0 & \cdots & 0 & 0 \\
        1 & x & \cdots & x & 0 \\
        & 1 & \ddots & \vdots & \vdots \\
        & & \ddots & x & 0 \\
        & & & 1 & 0
    \end{bmatrix}
    =
    \begin{bmatrix}
        0 & 0 \\
        \mathbf{U} & 0
    \end{bmatrix} \>,
\end{equation}
and
\begin{equation}
    \mathbf{A}_{n-1}\mathbf{A}_{n-1}^{*} - \mathbf{A}_{n-1}^{*}\mathbf{A}_{n-1} =
    \begin{bmatrix}
        -1 & & & & \\
        & 0 & & & \\
        & & \ddots & & \\
        & & & 0 & \\
        & & & & 1
    \end{bmatrix} \>.
\end{equation}
Since
\begin{equation}
    \mathbf{A}_{n-1}^{*} =
    \begin{bmatrix}
        0 & \mathbf{U}^{*} \\
        0 & 0
    \end{bmatrix}
\end{equation}
and
\begin{equation}
    \mathbf{A}_{n-1}\mathbf{A}_{n-1}^{*} =
    \begin{bmatrix}
        0 & 0 \\
        0 & \mathbf{U}\mathbf{U}^{*}
    \end{bmatrix}
\end{equation}
and
\begin{equation}
    \mathbf{A}_{n-1}^{*}\mathbf{A}_{n-1} =
    \begin{bmatrix}
        \mathbf{U}^{*}\mathbf{U} & 0 \\
        0 & 0
    \end{bmatrix} \>,
\end{equation}
\begin{equation}
    \begin{bmatrix}
        0 & 0 \\
        0 & \mathbf{U}\mathbf{U}^{*}
    \end{bmatrix}
    -
    \begin{bmatrix}
        \mathbf{U}^{*}\mathbf{U} & 0 \\
        0 & 0
    \end{bmatrix}
\end{equation}
must be diagonal. Therefore, the first row of $\mathbf{U}\mathbf{U}^{*}$ must be zero except for the first element.

\begin{remark}
    For $n = 4$, and $P = \{0, i, -i\}$ ($m = 2$) the following 4 matrices are normal:
    
    \begin{center}
    \begin{tabular}{|*1{>{\centering\arraybackslash}p{.05\textwidth}|}*2{>{\centering\arraybackslash}p{.3\textwidth}|}}
    \hline
    $w_j$ & $w_j$-skew symmetric & $w_j$-skew circulant\\ \hline
    \[i\] &
    \[\begin{bmatrix}
        0 & i & 0 & 0\\
        1 & 0 & i & 0\\
          & 1 & 0 & i\\
          &   & 1 & 0
    \end{bmatrix}\] & 
    \[\begin{bmatrix}
        0 & 0 & 0 & i\\
        1 & 0 & 0 & 0\\
          & 1 & 0 & 0\\
          &   & 1 & 0
    \end{bmatrix}\]\\ \hline
    \[-i\] &
    \[\begin{bmatrix}
        0 & -i & 0 & 0\\
        1 & 0 & -i & 0\\
          & 1 & 0 & -i\\
          &   & 1 & 0
    \end{bmatrix}\] & 
    \[\begin{bmatrix}
        0 & 0 & 0 & -i\\
        1 & 0 & 0 & 0\\
          & 1 & 0 & 0\\
          &   & 1 & 0
    \end{bmatrix}\]\\ \hline
    \end{tabular}
    \end{center}
\end{remark}

%
\section{Stable Matrices}
\newcommand{\A}{\ensuremath{\mathbf{A}}}
\subsection{Type I Stable Matrices}
A \textsl{Type I stable matrix} $\A$
is a matrix with all of its eigenvalues strictly in the left half plane: if $\lambda$ is an eigenvalue of~$A$ then $\Re(\lambda) < 0$. This nomenclature comes from differential equations, in that all solutions of the linear system of ODEs $dy/dt = \A y$ will ultimately decay as $t \to \infty$ if $\A$ is a type I stable matrix. 

If the matrix $\A$ is not \textsl{normal}, then \textsl{pseudospectra} can play a role, in that even though all solutions $y$ must ultimately decay, they might first grow large.  See~\cite{Embree:HLA:2013} for details.

By Theorem~\ref{thm:zero_diag_UH}, only $2m$ of the zero diagonal upper Hessenberg matrices with population $P = \{-1, 0, {+1}\}$ are normal, where here $m=2$.  Similarly, when the population is $P=\{0, {+1}\}$ then $m=1$ and only two matrices of every dimension are normal (the symmetric matrix with $1$s on its upper diagonal, and the circulant matrix with a $1$ in the last column of the first row).

\begin{theorem}
\label{thm:not_stable}
No $\mathbf{A}_n \in \ZH{\theta_k}{n}{P}$ is Type I stable, for any population $P$.
\end{theorem}

\begin{proof}
Suppose $\mathbf{A}_n \in \ZH{\theta_k}{n}{P}$ has eigenvalues $\{\lambda_{k}\}_{k=1}^{n}$. Then 
\begin{equation}
\sum_{k=1}^{n}\lambda_{k} = \mathrm{Trace}(\mathbf{A}_n) = 0\>. 
\end{equation}
Therefore, $\sum_{k=1}^{n}\mathrm{Re}(\lambda_{k}) = 0$. This is $n$ times the average, and so the average is zero.  Since the maximum $\mathrm{Re}(\lambda_{k})$ must be larger than the average, this proves the theorem.
\end{proof}
The proof of this theorem did not depend on the structure or population.  Thus if we consider
$\HH{0}{n}{P}$ instead of $\ZH{0}{n}{P}$, then we may simplify our search for stable matrices by restricting the computation to those with negative trace. This is in fact the first inequality of the Hurwitz criteria\footnote{The Maple command \texttt{PolynomialTools[Hurwitz]} implements a well-known test to decide if $p \in \mathbb{C}[z]$ has all its roots strictly in the left half plane.  Because that routine considers the complex case, and tests for pathological cases, it is too inefficient to use in this context.  We unrolled the loops, essentially converting the code to specific tests of the principal minors of the Hurwitz matrix.}, which leads to an effective and efficient method to count stable matrices: start from the database of characteristic polynomials~\cite{CPDB}, decide using the Hurwitz criteria if all roots are in the left half-plane, and if so add its matrices to the count.

\begin{table}[ht]
    \centering
    \begin{tabular}{|c|c|c|}
          \hline 
     $n$ & $\HH{0}{n}{\{-1, 0, {+1}\}}$  & $\HH{0}{n}{\{-1, {+1}\}}$  \\
        \hline 
    2   & 4  & 1 \\
          \hline 
    3   & 44 & 4\\
          \hline 
    4   & 1,386 & 28  \\
          \hline 
    5   & 130,735 & 424 \\
          \hline 
    6   & 35,217,156 &  11,613 \\
          \hline 
7   &  &  617,619 \\
          \hline 
    \end{tabular}
    \caption{The numbers of Type I stable matrices for various populations and dimensions. }
    \label{tab:stablesI}
\end{table}
\begin{figure}[h]
    \centering
    \includegraphics[width=\textwidth]{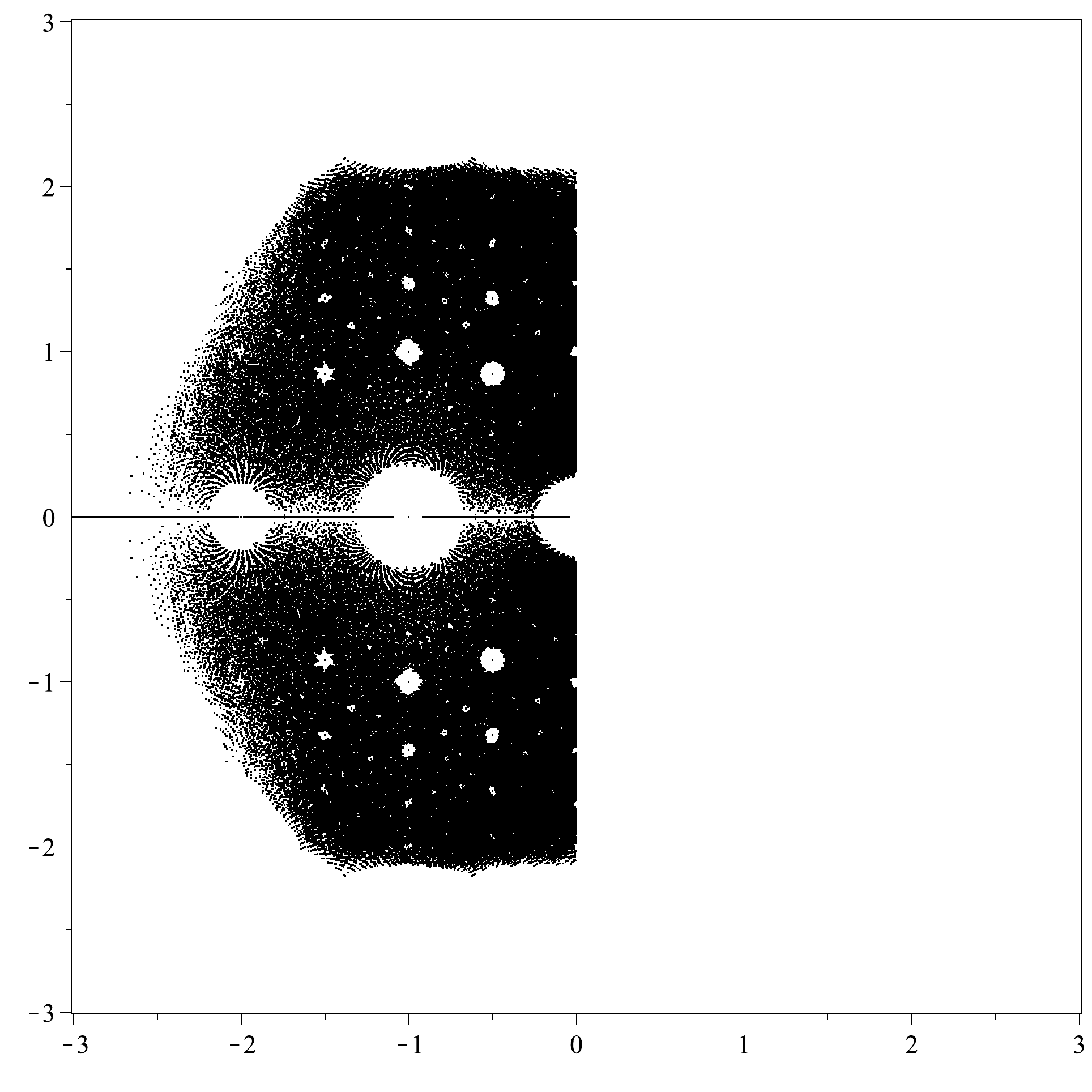}
    \caption{All eigenvalues of all 35,217,156 stable matrices from~\HH{0}{6}{\{-1, 0, {+1}\}}.  The maximum real part is approximately $-2.42\cdot 10^{-5}$.  There were only $55,298$ distinct characteristic polynomials from all these matrices.}
    \label{fig:stableHessenberg6}
\end{figure}

\begin{remark}
For stable matrices in $\HH{0}{n}{\{-1, 0, {+1}\}}$ the maximum real part of any eigenvalue is, for $n=2$, just $-0.5$ while for $n=3$ it is $-1.226\cdot 10^{-1}$.  For $n=4$ it is $-1.591\cdot10^{-2}$. For $n=5$ it is $-5.176\cdot 10^{-4}$.  For $n=6$ it is $-2.42\cdot 10^{-5}$.  The maximum real part of the eigenvalues seems to be approaching the real axis at least exponentially in $n$, for this population.  It would be nice to have a good asymptotic estimate.

The sequence of maximum real parts of eigenvalues for $\HH{0}{n}{\{-1, {+1}\}}$ gives at $n=2$ $\Re(\lambda)=-1$, $-0.5$, $-2.168\cdot 10^{-2}$, $-2.66\cdot 10^{-3}$, $-1.70\cdot10^{-4}$, and $-2.62\cdot10^{-6}$ for $n=7$.
\end{remark}

%
\subsection{Type II Stable matrices}
A \textsl{Type II Stable Matrix} $\A$ has all its eigenvalues inside the unit circle.  This class of matrices arises naturally on studying the simple linear recurrence relation $y_{n+1} = \A y_n$.
Fairly obviously, all solutions of this difference equation will ultimately decay to $0$ as $n \to \infty$ if and only if all eigenvalues of $\A$ are inside the unit circle (again, pseudospectra can play a role in the transient behaviour, sometimes significantly).

\begin{theorem}
If $\A$ is a Bohemian matrix with integer population $P$, then it is Type II stable if and only if it is nilpotent, in which case all its eigenvalues are $0$.
\end{theorem}
\begin{proof}
Suppose to the contrary that some eigenvalues are not zero.

The determinant of $\A$ must necessarily be an integer.  If the integer is not zero, it is at least $1$ in magnitude.  The product of the eigenvalues is thus at least $1$ in magnitude; hence there must be at least one eigenvalue that is at least $1$ in magnitude.

If the matrix $\A$ has zero determinant but not all eigenvalues zero, then after factoring out $z^m$ for the multiplicity of the zero eigenvalue, the product of the other eigenvalues becomes the constant coefficient (what was the coefficient of $z^m$ in the original).  This coefficient again must be an integer, and again at least one eigenvalue must be at least $1$ in magnitude.

This proves the theorem, by contradiction.
\end{proof}

\begin{remark}
We did not, in fact, use that the matrix came from a Bohemian family; only that its entries were integers.
\end{remark}

Searching for nilpotent matrices in various classes of Bohemian matrices turns up several puzzles. We give some preliminary results here in Table~\ref{tab:nilpotents}, but leave this mostly to future work.  For instance, it seems clear from our experiments that the only nilpotent matrix in $\HH{0}{n}{\{0, {+1}\}}$ is the (transpose of the) complete Jordan block of $n$ zero eigenvalues; contrariwise the irregular behaviour for $\HH{0}{n}{\{-1, {+1}\}}$ is very puzzling.

\begin{table}[ht]
    \centering
    \begin{tabular}{|c|c|c|c|}
          \hline 
     $n$ & $\ZH{0}{n}{\{-1, 0, {+1}\}}$ & $\HH{0}{n}{\{0, {+1}\}}$ & $\HH{0}{n}{\{-1, {+1}\}}$  \\
        \hline 
    2   & 1 & 1 & 2 \\
          \hline 
    3   & 3& 1 & 0 \\
          \hline 
    4   & 21 & 1 & 0 \\
          \hline 
    5   & 271 & 1 & 0 \\
          \hline 
    6   & 9,075 & 1 & 324 \\
          \hline 
    \end{tabular}
    \caption{The numbers of nilpotent matrices for various populations and dimensions}
    \label{tab:nilpotents}
\end{table}

%
Considering \textsl{general} Bohemian matrices with population $\{-1, 0, {+1}\}$, so that there are $3^{n^2}$ such matrices, we find that there are $1$, $9$, $481$, $148,817$, and $243,782,721$ nilpotent matrices at dimensions $1$ through $5$ inclusive.  We can fit this experimentally with the formula $\exp(0.5 + 0.38n + 0.23n^2)$, or something like $1.26^{n^2}$, which vanishes very quickly compared to $3^{n^2}$. This formula predicts that for $n=6$ the probability of finding a nilpotent matrix is about $2.75\times 10^{-14}$.  It would be gratifying to have a better understanding of the number of nilpotent matrices in a family.



\section{Concluding Remarks}
The class of upper Hessenberg Bohemian matrices gives a useful way to study Bohemian matrices in general. This is an instance of Polya's adage ``find a useful specialization."~\cite[p.~190]{polya2014solve} Because these classes are simpler than the general case, we were able to establish several theorems. Note that the three families $\HH{0}{n}{\{0, {+1}\}}$, $\HH{0}{n}{\{-1, {+1}\}}$, and $\ZH{0}{n}{\{-1, 0, {+1}\}}$ are all subfamilies of $\HH{0}{n}{\{-1, 0, {+1}\}}$.

In this paper we have introduced two new formulae for computing the characteristic polynomials of upper Hessenberg matrices. Our first formula, given in Theorem~\ref{thm:charPolyRec1}, also computes the characteristic polynomials recursively. Our second formula, given in Theorem~\ref{thm:charPolyRec2}, computes the coefficients recursively. We also explored some properties of zero diagonal Bohemian upper Hessenberg matrices. In Theorem \ref{thm:zero_diag_UH}, we show that the subset of these matrices that are normal are always symmetric, $w_{j}$-skew symmetric for some fixed $1 \leq j \leq m$, or $w_{j}$-skew circulant. In Theorem \ref{thm:not_stable}, we showed that no $\mathbf{H} \in \ZH{\theta_{k}}{n}{P}$ is stable.

Many puzzles remain. Perhaps the most striking is the angular appearance of the set $\mathbf{\Lambda}(\HH{0}{n}{P})$ of eigenvalues of $\HH{0}{n}{P}$, such as in Figures~\ref{fig:UH_6} and \ref{fig:UH_6_0_Diag}. General matrices have eigenvalues asymptotic to a (scaled) disc~\cite{tao2017random}; our computations suggest that as $n \to \infty$, $\sfrac{\mathbf{\Lambda}(\HH{0}{n}{P})}{n^{\sfrac{1}{2}}}$ tends to an irregular hexagonal shape, rather than a disk. More, the density does not appear to be approaching uniformity. Further, the boundary is irregular, with shapes suggestive of what is popularly known as the ``dragon curve" (in reverse---these delineate where the eigenvalues are absent, near the edge). We have no explanation for this.

\section*{Acknowledgements}
The calculations and images presented here were in part made possible using AMD Threadripper workstations provided by the Department of Applied Mathematics at Western University.
We acknowledge the support of the Ontario Graduate Institution, The National Science \& Engineering Research Council of Canada, the University of
Alcal\'a, the Rotman Institute of Philosophy, the Ontario Research Centre of
Computer Algebra, and Western University. Part of this work was developed
while R.~M.~Corless was visiting the University of Alcal\'a, in the frame of the
project Giner de los Rios. L.~Gonzalez-Vega, J.~R.~Sendra and J.~Sendra are
partially supported by the Spanish Ministerio de Econom\'\i a y Competitividad
under the Project MTM2017-88796-P.

\bibliographystyle{siamplain}  
\bibliography{bibliography}

\end{document}